\documentclass[12pt, draftclsnofoot, onecolumn]{IEEEtran}
\usepackage{cite}
\usepackage{graphicx}
\usepackage{amsmath}
\usepackage{amsfonts}
\usepackage{balance}
\usepackage{fancyhdr}
\usepackage{graphicx,times,cite,amssymb,epsfig,amsthm,color}
\newtheorem{theorem}{Theorem}

\newlength{\figwidth}
\normalsize

\usepackage{graphicx}
\ifCLASSINFOpdf
\else
\fi
\usepackage{amsmath}
\DeclareMathOperator*{\argmax}{argmax}

\usepackage[cmintegrals]{newtxmath}
\usepackage{algorithm}
\usepackage{algorithmic}

\usepackage{stfloats}
\usepackage{url}

\usepackage{verbatim}

\makeatletter
\newcommand*{\rom}[1]{\expandafter\@slowromancap\romannumeral #1@}
\makeatother

\newcommand{\dis}{d}
\newcommand{\dhalf}{m}
\newcommand{\dspec}{A_i(\dis)}

\newcommand{\dset}{\mathcal{D}}
\newcommand{\rcvd}{\boldsymbol{r}}
\newcommand*\diff{\mathop{}\!\mathrm{d}}
\newcommand{\numcodes}{M}
\newcommand{\dimension}{K}
\newcommand{\signal}{\boldsymbol{s}}

\begin{document}
%
\title{A Modified Union Bound on Symbol Error Probability for Fading Channels}
%
%
%

\author{Tian~Han,
Rajitha~Senanayake,
Peter~Smith,
Jamie~Evans}
\maketitle

\begin{abstract}
In this paper, we propose a new upper bound on the error probability performance of maximum-likelihood (ML) detection. The proposed approach provides a much tighter upper bound when compared to the traditionally used union bound, especially when the number of pairwise error probabilities (PEPs) is large. In fact, the proposed approach tightens the union bound by first assuming that a detection error always occurs in a deep fading event where the channel gain is lower than a certain threshold. A minimisation is then taken with respect to the gain threshold in order to make the upper bound as tight as possible. We also prove that the objective function has a single minimiser under several general assumptions so that the minimiser can be easily found using optimisation algorithms. The expression of the new upper bound under correlated Rayleigh fading channels is derived and several analytical and numerical examples are provided to show the performance of the proposed bound.

\end{abstract}


%
\IEEEpeerreviewmaketitle

\section{Introduction}
%
%
%
%
\IEEEPARstart{I}{n} wireless communication systems, error probability performance is used as an important measure of the system reliability. The analytical derivation of such performance measures is very important as it provides insights and a reduction in computation time compared to Monte Carlo simulations. However, the calculation of the exact error probability is a very challenging problem, especially for communication systems with complex constellations and, to the best of our knowledge, is still an open problem. The difficulty lies in analyzing the complex decision regions formed in a general constellation. 
In the past, research has focused on deriving upper and lower bounds in order to approximate the error probability performance. One simple upper bond on the error probability is the union bound, which makes use of Boole's inequality. 
While it is commonly used in the literature due to its low computational complexity \cite{Candreva10}, it can significantly overestimate the exact error probability, as it adds pairwise error probabilities (PEPs) up independently. In fact, when the number of pairwise error events is large, it can even result in a bound greater than $1$ in the low signal-to-noise ratio (SNR) region.

Some research has focused on looking for tighter upper bounds based on the union bound under particular scenarios. In \cite{Ma13, Liu20}, techniques for upper bounding the error probability of linear block codes are presented making use of the truncated weighted spectrum. Gallager's first bounding technique (GFBT) is considered in both \cite{Ma13} and \cite{Liu20} in order to improve the union bound in the low SNR region. In \cite{Mouchtak16}, tight upper bounds on the error probability performance of convolutional codes over correlated Rayleigh fading channels are derived based on analytical upper bounds for the Gaussian Q-function and a lower bound on the channel fading covariance matrix. 
Upper bounds on the performance of linear block codes under maximum likelihood decoding operating over exponentially correlated generalised-fading channels are proposed in \cite{Mouchtak19}. Particularly, the analysis considers three different fading scenarios, namely Rice, Nakagami-m, and Weibull fading. In \cite{Luong15}, a new upper bound on the bit error probability of high-rate spatial modulation system using quadrature amplitude modulation (QAM) over a Rayleigh fading channel is presented. The closed-form expression is derived by eliminating unnecessary PEPs and applying Verdu's theorem \cite{Verdu87}. In \cite{Yoon20}, a simplified and more tractable bound for evaluating the performance of dual carrier modulation (DCM) over a Nakagami-m fading channel is presented by using the generic formula in \cite{Park06}. 

In this paper, we consider a signalling scheme with a general constellation and propose an improved upper bound on the error probability performance. The proposed result is derived for any general fading channel model and can be applied to a vast range of examples that involve a large number of pairwise error events resulting in poor union bounds. Our approach takes into account the effect of the deep fading event on the error probability performance. The resulting upper bound is effective not only in the high SNR region, but also in the very low SNR region. We present several interesting numerical examples to illustrate the applicability of our results to a range of scenarios.

The notations used throughout this paper are given as follows. We use upper and lower case bold symbol letters to represent matrices and  vectors, respectively. For a vector $\boldsymbol{v}$, we use $\| \boldsymbol{v}\|$ to represent its $L^2$-norm. We use $\boldsymbol{I}_P$ to represent a $P \times P$ identity matrix and diag$(x_1, \, \cdots, \, x_P)$ to represent a $P \times P$ diagonal matrix with diagonal elements $x_1, \, \cdots, \, x_P$.
We use calligraphic letters to represent sets and $\mathbb{R}$ and $\mathbb{C}$ to represent the sets of real and complex numbers, respectively. For a real number $x \in \mathbb{R}$, we use $\lfloor x \rfloor$ to represent the operation of rounding $x$ down to the largest integer less than or equal to $x$. For a complex number $c \in \mathbb{C}$, we denote $\Re \{c \}$ and $|c|$ as its real part and its absolute value, respectively. For a square matrix $\boldsymbol{C} \in \mathbb{C}^{P\times P}$, we use $\boldsymbol{C}^T$ and $\boldsymbol{C}^H$ to represent its transpose and Hermitian, respectively.
We use $Q(\cdot)$ to represent the Gaussian Q-function. We use $A \equiv B$ to represent the statistical equivalence between the two random variables $A$ and $B$. For a real valued random vector $\boldsymbol{v} \in \mathbb{R}^{P \times 1}$, we use $\boldsymbol{v} \sim \mathcal{N}(\boldsymbol{\mu}, \boldsymbol{\Sigma})$ to denote that $\boldsymbol{v}$ follows a $P$-dimensional Gaussian distribution where $\boldsymbol{\mu}\in \mathbb{R}^{P \times 1}$ and $\boldsymbol{\Sigma}\in \mathbb{R}^{P \times P}$ are its mean vector and covariance matrix, respectively. For a complex valued random vector $\boldsymbol{v} \in \mathbb{C}^{P \times 1}$, we use $\boldsymbol{v} \sim \mathcal{CN}(\boldsymbol{0}, \boldsymbol{\Sigma})$ to represent that $\boldsymbol{v}$ follows a $P$-dimensional circularly symmetric complex zero mean Gaussian distribution where $\boldsymbol{\Sigma}\in \mathbb{C}^{P \times P}$ represents the covariance matrix. 



The rest of this paper is organised as follows. In section \ref{sec:formulation}, we provide the system model and derive the traditional union bound. In section \ref{sec:UB}, we propose the general expression for the new upper bound and prove a property of the bound. We also derive the expression over a Rayleigh fading channel. In section \ref{sec:examples}, we provide analytical and numerical examples to analyse the performance of the proposed bound. In section \ref{sec:conclusion}, we make conclusions and discuss some possible future extensions. 

\section{Problem Formulation}   \label{sec:formulation}
\subsection{System Model} 
Let us consider a general complex valued signalling scheme of dimension $\dimension$, in which the real part and the imaginary part denote the in-phase and quadrature components, respectively. Assume that there are $\numcodes$ signals in this signalling scheme, denoted as $\signal_1, \, \cdots, \,\signal_{\numcodes}$. The constellation is normalised such that the average energy of the scheme $\frac{1}{\numcodes}\sum_{i=1}^{\numcodes} \|\signal_i\|^2 = 1$. For example with Quadrature phase shift keying (QPSK), we have $K = 1$, $M = 4$ and $\{s_1 = 1, \,s_2 = j, \, s_3 = -1, \,s_4 = -j\}$. 
The distance spectrum, denoted as $\dspec$, $\dis \in \dset \subset \mathbb{R}$, is defined as the number of signals with Euclidean distance $\dis$ to the $i$-th signal, where $\dset$ is the set of all possible distances between two signals. Note that the distance spectrum is dependent on the particular signal chosen.


Let us consider a system with a single transmit antenna and $N$ receive antennas. Under the block fading channel model, the $N$ channel fading factors remain unchanged during the transmission of each signal. Assuming the transmitted signals are equally likely, the received signal of such a system can be represented by an $N\dimension \times 1$ vector as
\begin{equation}
\rcvd = \sqrt{E} \boldsymbol{H}  \signal_i + \boldsymbol{n},
\end{equation}
where
\begin{equation}
    \boldsymbol{H} = \begin{bmatrix}
    h_1 \boldsymbol{I}_{\dimension} \\
    \vdots \\
    h_N \boldsymbol{I}_{\dimension}
\end{bmatrix},
\end{equation}
with $h_n$ representing the $n$-th complex channel fading factor $\forall n \in \{1,\,\cdots,\, N \}$, $E$ representing the average energy contained in all possible transmitted signals
, $\sqrt{E}\boldsymbol{s}_i, \, \boldsymbol{s}_i \in \{\boldsymbol{s}_1, \, \cdots, \, \boldsymbol{s}_{\numcodes}\}$ representing the transmitted signal and $\boldsymbol{n} \sim \mathcal{CN}(0, \sigma^2 \boldsymbol{I}_{N\dimension})$ representing the complex additive white Gaussian noise (AWGN) vector. Assuming that perfect channel state information is available at the receiver, the maximum likelihood (ML) detection of the received signal, denoted by $\hat{\boldsymbol{s}}_i$, is given by
\begin{equation}
\hat{\signal}_i = \argmax_{\signal_k\in \{\signal_1,...,\signal_{\numcodes}\}} 2 \Re\left\{\boldsymbol{r}^H \boldsymbol{H} \signal_k\right\} -  \sqrt{E} \|\signal_k \|^2 \sum_{n=1}^N |h_n|^2  , \label{eq:ML}
\end{equation}
where the maximisation is over all possible transmitted signals.

\subsection{Error Probability Analysis}
As mentioned earlier, the exact calculation of the error probability performance for a general constellation is still an open problem. While the exact analytical expression may be intractable due to the complex form of the decision regions, several upper and lower bounds have been proposed in the literature to approximate the error probability performance. From among them, the union bound is a widely used upper bound which leads to a simpler expression. For the proposed signalling scheme and ML detection, the union bound on the block error probability can be expressed as
\begin{equation}
    P_e \leq P_e^{UB} = \frac{1}{\numcodes}\sum_{i=1}^{\numcodes}\sum_{\dis \in \dset} \dspec P_{ik_{\dis}}, \label{eq:UB}
\end{equation}
where $P_e$ is the exact block error probability, $P_e^{UB}$ is the union bound, $P_{ik_{\dis}}$ denotes the 
PEP of preferring the $k_d$-th signal when the $i$-th signal is transmitted, which can be expressed as
\begin{equation}
P_{ik_{\dis}} = \Pr\left[\xi_{ii}<\xi_{ik_{\dis}} \right], \label{eq:pik}
\end{equation}
where $\xi_{ik}$ is the $k$-th decision variable when the $i$-th signal is transmitted, which is defined as
\begin{equation}
\xi_{ik} = 2 \Re\left\{\boldsymbol{r}^H \boldsymbol{H} \signal_k\right\} - \sqrt{E}\|\signal_k \|^2\sum_{n=1}^N |h_n|^2 . \label{eq:dv}
\end{equation}
Note that the Euclidean distance between the $k_d$-th signal and the $i$-th signal is $d$.   

Conditioned on the channel fading, \eqref{eq:pik} can be re-expressed as
\begin{equation}
    P_{ik_{\dis}} = \int_{0}^{\infty} \Pr\left[\xi_{ii}<\xi_{ik_{\dis}} \mid X=x \right]f_X(x)\diff x, \label{eq:pik1}
\end{equation}
where $f_X(x)$ is the probability density function (pdf) of the fading channel gain, $X$, which is defined as $X \equiv \boldsymbol{h}^H\boldsymbol{h}$, with $\boldsymbol{h}$ representing the channel fading vector, $\boldsymbol{h} = [h_1,\, \cdots,\, h_N]^T \in \mathbb{C}^{N \times 1}$.
Substituting \eqref{eq:dv} into \eqref{eq:pik1} and doing some mathematical manipulations we can reexpress the PEP as
\begin{equation}
P_{ik_{\dis}} = \int_{0}^{\infty} Q\left( \dis\sqrt{\frac{x}{2}\frac{E}{\sigma^2}} \right) f_X(x)\diff x. \label{eq:conditioned_pik2}
\end{equation}
Substituting \eqref{eq:conditioned_pik2} into \eqref{eq:UB}, the union bound for a general complex constellation can be expressed as
\begin{equation}
P_e^{UB} = \frac{1}{\numcodes}\sum_{i=1}^{\numcodes}\sum_{\dis \in \dset} \dspec \int_{0}^{\infty} Q\left( \dis\sqrt{\frac{x}{2}\frac{E}{\sigma^2}} \right) f_X(x)\diff x. \label{eq:UB_fin}
\end{equation}

While the union bound has been widely used in the literature due to its simplicity, it over-estimates the error probability performance and contains no information when its value is above 1. This is more prominent in applications with a large number of signals as here the union bound accumulates a large number of PEPs making the upper bound quite loose. 
Furthermore, the union bound might be loose in fading channels due to the existence of deep fading events. When the channel is in a deep fade, the instantaneous received SNR is low. This results in large PEPs, making the union bound very large and uninformative.

Motivated by these limitations, in the following we propose a new upper bound that is tighter than the union bound in block fading channels.

\section{The New Upper bound} \label{sec:UB}
Let us first define the deep fading event as an event in which the instantaneous channel gain, $X$, falls below a predefined threshold, i.e., $X < \gamma$. 
In order to deal with such events, we propose a new upper bound by assuming that a detection error always occurs when the channel gain is lower than $\gamma$. Thus, the error probability expression can be evaluated separately for $X < \gamma$ and $X \geq \gamma$ to produce the new upper bound as given below 
\begin{equation}
    P_e^{NB} = \min_{\gamma \geq 0} \left\{\Pr \left[X <\gamma \right] + \Pr \left[X \geq \gamma \right]\frac{1}{\numcodes}\sum_{i=1}^{\numcodes}\sum_{\dis \in \dset} \dspec \Tilde{P}_{ik_{\dis}} \right\}, \label{eq:newUB}
\end{equation}
where the minimisation is taken with respect to $\gamma$ in order to make the upper bound as tight as possible \cite{han2021} 
and $\Tilde{P}_{ik_{\dis}}$ is the PEP conditioned on $X \geq \gamma$ which can be expressed as
\begin{equation}
\Tilde{P}_{ik_{\dis}} = \Pr\left[\xi_{ii}<\xi_{ik_{\dis}} \mid X \geq \gamma \right]. \label{eq:conditioned_pik}
\end{equation}
Reexpressing \eqref{eq:conditioned_pik} as
\begin{equation}
    \Tilde{P}_{ik_{\dis}} = \int_{\gamma}^{\infty} \Pr\left[\xi_{ii}<\xi_{ik_{\dis}} \mid X=x \right]\frac{f_X(x)}{\Pr[X \geq \gamma]}\diff x, \label{eq:conditioned_pik1}
\end{equation}
and substituting \eqref{eq:conditioned_pik1} into \eqref{eq:newUB} we obtain the new upper bound on the block error probability performance as
\begin{equation}
    P_e^{NB} = \min_{\gamma \geq 0} \left\{\Pr \left[X<\gamma \right] + \frac{1}{\numcodes}\sum_{i=1}^{\numcodes}\sum_{\dis \in \dset} \dspec \int_{\gamma}^{\infty} Q\left( \dis\sqrt{\frac{x}{2}\frac{E}{\sigma^2}} \right) f_X(x)\diff x \right\}. \label{eq:newUB_fin}
\end{equation} 
Note that the extra computational complexity required to calculate the new bound compared to \eqref{eq:UB_fin} is mainly introduced due to the minimisation procedure. In the following, we prove that the objective function to be minimised in \eqref{eq:newUB_fin} has a single minimum under some very general assumptions. Hence, the optimisation problem can be solved very efficiently using simple algorithms.


\begin{theorem} \label{theo:property}
Consider a general complex valued signalling scheme of dimension $\dimension$ in which there are $\numcodes \geq 3$ signals, $s_i(t), i \in \{1,\,\cdots,\, \numcodes\}$. Assume that the constellation of this signaling scheme is normalised such that $\frac{1}{\numcodes}\sum_{i=1}^{\numcodes} \|\signal_i\|^2 = 1$. $X$ is a continuous random variable representing the channel gain which is defined in $[0, \, \infty)$. Assume its pdf, denoted by $f_X(x)$, is differentiable, with $f_X(x) = 0,  x \in (-\infty, 0)$, $f_X(0) \geq 0$, and $f_X(x) > 0, x \in (0, +\infty)$. Then, the objective function defined as 
\begin{equation}
    G(\gamma) = \Pr \left[X<\gamma \right] + \frac{1}{\numcodes}\sum_{i=1}^{\numcodes}\sum_{\dis \in \dset} \dspec \int_{\gamma}^{\infty} Q\left( \dis\sqrt{\frac{x}{2}\frac{E}{\sigma^2}} \right) f_X(x)\diff x, \; \gamma \in [0, \, \infty),
\end{equation}
has a single minimiser. Note that $\dspec$, $\dis \in \dset \subset \mathbb{R}$, is the distance spectrum defined as the number of signals with Euclidean distance $\dis$ to the $i$-th signal in the normalised constellation and $\dset$ is the set of all possible Euclidean distances between each two signals. Besides, $E$ is the average transmitted signal energy and $\sigma^2$ is the variance of each element of the AWGN vector.
\end{theorem}
\begin{proof} 
See Appendix \ref{app:proof}.
\end{proof}

While \eqref{eq:newUB_fin} is valid for any fading channel model, in this section we derive an expression for the new upper bound for the special case of correlated Rayleigh fading channels. In this scenario, the channel fading vector, $\boldsymbol{h}$, can be written as
\begin{equation}
    \boldsymbol{h} \equiv \boldsymbol{R}^{1/2}\boldsymbol{u}, \label{eq:h_Ray}
\end{equation}
where $\boldsymbol{u} \sim \mathcal{CN}(0,\boldsymbol{I}_N)$ and $\boldsymbol{R}$ denotes the $N\times N$ correlation matrix of $\boldsymbol{h}$. As a Hermitian matrix, $\boldsymbol{R}$ can be expressed as
\begin{equation}
\boldsymbol{R} = \boldsymbol{V \Omega V}^H,  \label{eq:eig_decomp_Ru}
\end{equation}
where $\boldsymbol{V}$ is an $N \times N$ unitary matrix and $\boldsymbol{\Omega}=\text{diag}(\lambda_1, ..., \lambda_N)$ is a diagonal matrix containing the eigenvalues $\lambda_1, ..., \lambda_N$ of $\boldsymbol{R}$. As such, the random variable $X \equiv \boldsymbol{h}^H\boldsymbol{h}$ can be rewritten as
\begin{equation}
    X = \boldsymbol{u}^H \boldsymbol{R} \boldsymbol{u} \equiv \sum_{j=1}^N \lambda_j |w_j|^2,  \label{eq:Ray_hHh}
\end{equation}
where the $w_j \sim \mathcal{CN}(0,1), j \in \{1,\,2,\,\cdots,\,N \}$, are statistically identical to the elements of $\boldsymbol{u}$. Thus, the pdf of $X$ can be expressed as \cite{johnson1970}
\begin{equation}
\begin{aligned}
f_X(x) = \left\{\begin{aligned}
&\sum_{j=1}^N \frac{b_j}{\lambda_j} e^{-x/\lambda_j}, &\quad x \geq 0\\
&0, &\quad x < 0,
\end{aligned}
\right.  
\end{aligned} \label{eq:pdf_ray_hHh}
\end{equation}
where $b_j = \lambda_j^{N-1} \prod_{\substack{n=1 \\ n \neq j}}^N 1/(\lambda_j-\lambda_n)$. By integrating \eqref{eq:pdf_ray_hHh}, the cumulative density function can be obtained as below
\begin{equation}
\begin{aligned}
F_X(x) = \left\{\begin{aligned}
&\sum_{j=1}^N b_j (1-e^{-x/\lambda_j}), &\quad x \geq 0\\
&0, &\quad x < 0.
\end{aligned}
\right.  
\end{aligned} \label{eq:cdf_ray_hHh}
\end{equation}
Substituting \eqref{eq:pdf_ray_hHh} and \eqref{eq:cdf_ray_hHh} into \eqref{eq:newUB_fin}, the new upper bound in a correlated Rayleigh fading  channel can be written as
\begin{equation}
    P_e^{NB} = \min_{\gamma \geq 0} \Biggl\{\sum_{j=1}^N b_j (1-e^{-\gamma/\lambda_j}) + \frac{1}{\numcodes}\sum_{i=1}^{\numcodes}\sum_{\dis \in \dset}\sum_{j=1}^N \dspec \int_{\gamma}^{\infty} Q\left( \dis\sqrt{\frac{x}{2}\frac{E}{\sigma^2}} \right) \frac{b_j}{\lambda_j} e^{-x/\lambda_j}\diff x \Biggl\}. \label{eq:newUB2}
\end{equation}
In order to simplify \eqref{eq:newUB2}, we apply Craig's formula  and do some mathematical manipulations to derive 
\begin{equation}
\begin{aligned}
    P_e^{NB} = \min_{\gamma \geq 0} \Biggl\{\sum_{j=1}^N b_j (1-e^{-\gamma/\lambda_j}) + \frac{1}{\numcodes \pi}\sum_{i=1}^{\numcodes}\sum_{\dis \in \dset}\sum_{j=1}^N \dspec 
   &\int_{0}^{\pi/2} \frac{4 \sigma^2 b_j \sin^2 \theta}{d^2 E \lambda_j + 4 \sigma^2 \sin^2 \theta}\\
   &  \times  \exp{\left( -\frac{d^2 E \lambda_j + 4 \sigma^2 \sin^2 \theta}{4 \lambda_j \sigma^2 \sin^2 \theta}\gamma\right)}\diff \theta \Biggl\}. \label{eq:newUB1}
\end{aligned}
\end{equation}

\section{Analytic examples and numerical results}  \label{sec:examples}
It is important to note that the proposed bound is presented in a general form and applies to a wide range of scenarios. To highlight its application, in the following we present three interesting examples related to orthogonal signalling, permutations based coding and Gaussian distribution based random coding. 
The channel fading we consider in all three examples is Rayleigh fading. %
\subsection{Orthogonal signalling}
We first consider orthogonal signalling of dimension $\dimension$ \cite{Proakis08}, in order to present how the proposed bound performs in a widely used signalling scheme. The number of signals, $\numcodes$, in this scheme is always the same as the dimension $\dimension$. Since the distance between each two signals is the same, the distance spectrum can be written as
\begin{equation}
    \dspec = \numcodes-1, \quad \dis = \sqrt{2}. \label{eq:dis_spec_orth}
\end{equation}

Substituting \eqref{eq:dis_spec_orth} into \eqref{eq:newUB1} and noticing that the distance spectrum is independent of the particular choice of signal, the new upper bound on the error probability can be expressed as
\begin{equation}
\begin{aligned}
P_e^{NB} = \min_{\gamma \geq 0} \Biggl\{\sum_{j=1}^N b_j (1-e^{-\gamma/\lambda_j}) + \frac{\numcodes -1}{\pi}  \sum_{j=1}^{N} &\int_{0}^{\pi/2}\frac{2b_j L \sigma^2 \sin^2\theta}{ E\lambda_j + 2 L \sigma^2 \sin^2\theta}\\ 
& \times\exp{\left(\frac{ E\lambda_j + 2 L \sigma^2 \sin^2\theta}{2\lambda_j L \sigma^2 \sin^2\theta}\gamma \right)}\diff \theta \Biggl\}. \label{eq:newUB_FSK}
\end{aligned}
\end{equation}

Fig. \ref{fig:FSK} shows the block error probability performance versus the average received SNR for orthogonal signalling with $\numcodes \in \{16,\, 512\}$. We adopt the correlated Rayleigh fading channel model and set the number of receive antennas as $N = 2$ and the correlation coefficient as $\rho = 0.1$. For the channel correlation we consider the simple exponential model. Hence, given a correlation coefficient $\rho \in [0,1]$, the $(i,j)$-th entry of $\boldsymbol{R}$ can be written as
\begin{equation}
\boldsymbol{R}(i,j) = \rho^{|i-j|}.  \label{eq:corrmat}
\end{equation}
Note that the correlation model in \eqref{eq:corrmat} will be used in the remainder of the section. 

From Fig. \ref{fig:FSK}, we observe that the new upper bound given in \eqref{eq:newUB_FSK} always outperforms the traditional union bound. We also observe that in the high SNR regime, the SNR gap between the new upper bound and the union bound is 0.5dB for $\numcodes = 16$ and 4.4dB for $\numcodes = 512$. This indicates that the effect of the significant increase in the number of PEPs with $\numcodes$ is more prominent in the union bound. 
\begin{figure}[t]
    \centerline{\includegraphics[width=10cm,height=7.8cm]{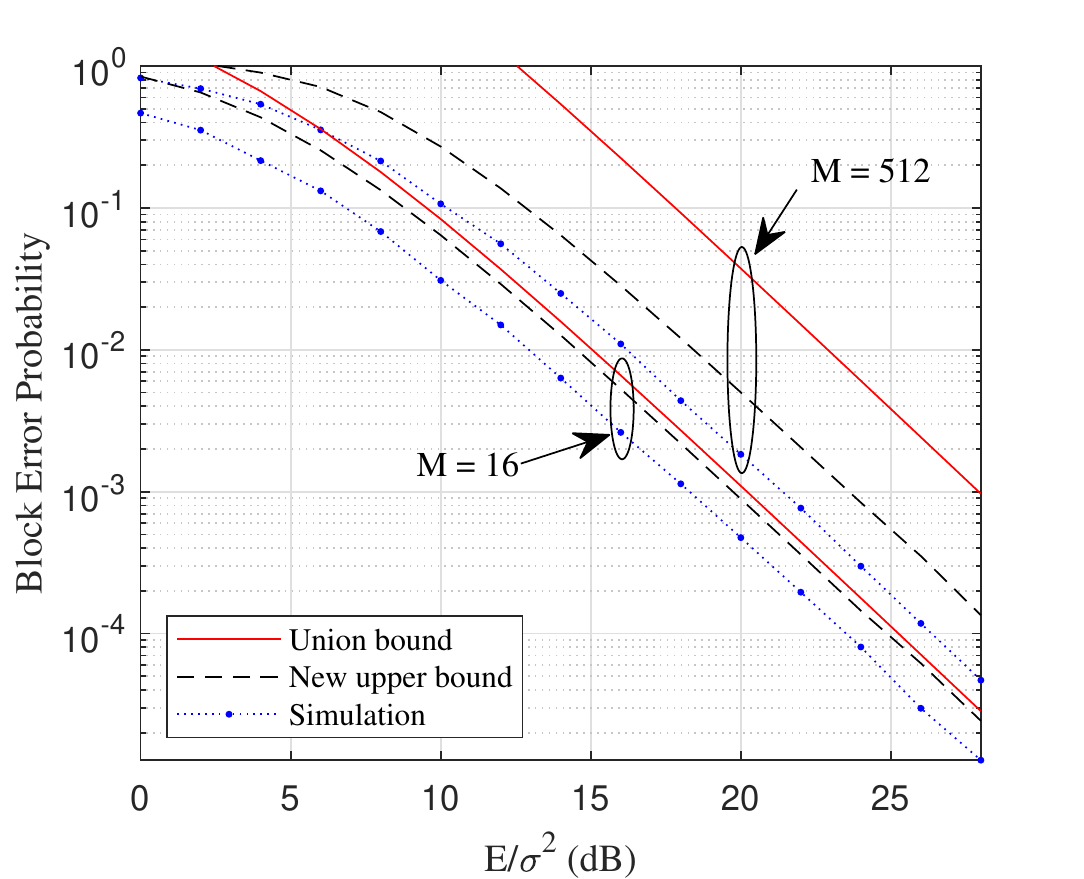}}
    \caption{The block error probability versus average received SNR for orthogonal signalling with $\numcodes \in \{16,\, 512\}$ in a correlated Rayleigh fading channel with $N = 2$ and $\rho = 0.1$.
    }\label{fig:FSK}
\end{figure}

In order to illustrate the impact of channel parameters on the proposed bound, Fig. \ref{fig:FSK_N4_rho01} and Fig. \ref{fig:FSK_N4_rho05} show the block error probability performance versus the average received SNR for orthogonal signalling with $\numcodes \in \{16,\, 512\}$ under the correlated Rayleigh fading channel model with $N = 4$ and $\rho = 0.1, 0.5$. Compared with Fig. \ref{fig:FSK} in which $N=4$, the SNR gap between the proposed bound and the union bound in Fig. \ref{fig:FSK_N4_rho01} decreases to approximately 0dB when $M=16$, and 1dB when $M=512$. Though the gap still increases with the number of PEPs, this result indicates that the improvement of the tightness is degraded by increasing the number of antennas. The SNR gaps in Fig. \ref{fig:FSK_N4_rho01} and Fig. \ref{fig:FSK_N4_rho05} are similar, which shows that the change of the correlation coefficient has little impact on the effectiveness of the proposed bound.
\begin{figure}[t]
    \centerline{\includegraphics[width=10cm,height=7.8cm]{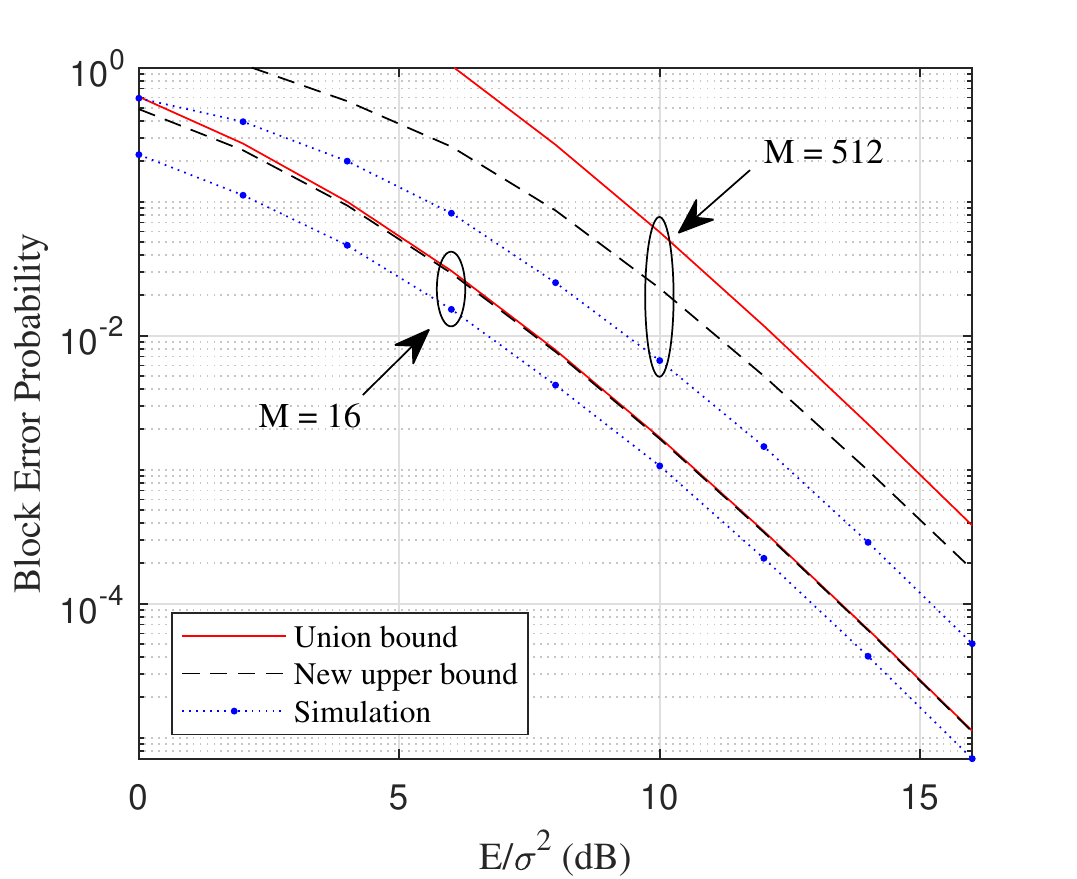}}
    \caption{The block error probability versus average received SNR for orthogonal signalling with $\numcodes \in \{16,\, 512\}$ in a correlated Rayleigh fading channel with $N = 4$ and $\rho = 0.1$.
    }\label{fig:FSK_N4_rho01}
\end{figure}
\begin{figure}[t]
    \centerline{\includegraphics[width=10cm,height=7.8cm]{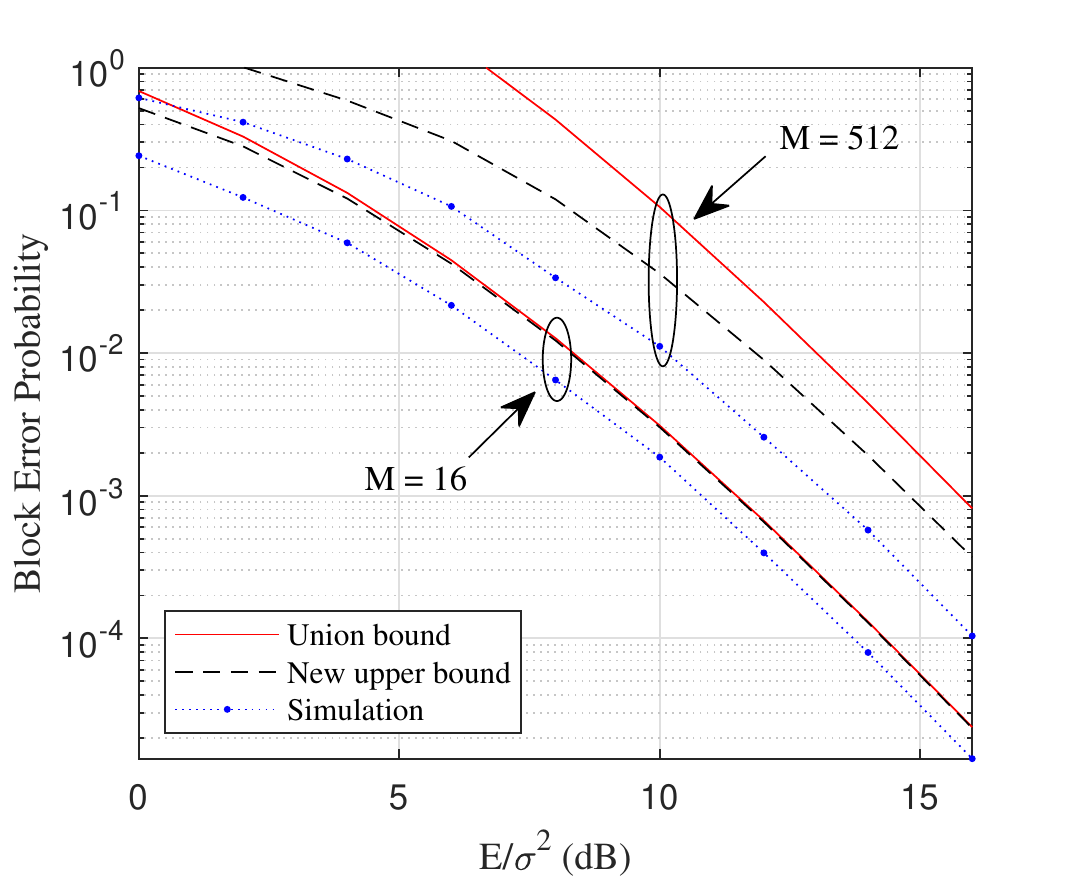}}
    \caption{The block error probability versus average received SNR for orthogonal signalling with $\numcodes \in \{16,\, 512\}$ in a correlated Rayleigh fading channel with $N = 4$ and $\rho = 0.5$.
    }\label{fig:FSK_N4_rho05}
\end{figure}

\subsection{Permutations based coding} \label{sec:UBperm}
In this example we present another topical application of the proposed bound in the area of integrated radar and communications. In \cite{Senanayake22TWC}, a new method is proposed for using permutations based random stepped frequency radar waveforms to enable joint radar and communications. Error probability performance analysis of the proposed waveform results in a large number of PEPs making the proposed upper bound a better candidate than the traditional union bound.  

Let us consider permutations based random stepped frequency radar waveforms with $L$ subpulses. We first convert the waveform into a binary code. For a waveform with $L$ subpulses, the dimension of the corresponding binary code is $\dimension = L^2$. The signal can be divided into $L$ slots, each representing a subpulse. If the frequency in the $l$-th subpulse is $f_n$, where $n\in\{0,\,1,\,...,\,L-1\}$, the $(n+1)$-th element in the $l$-th slot is one, while the other elements in this slot are 0's. As such, each signal has exactly $L$ $1$'s. 

The distance spectrum of the proposed code, regardless of the particular codeword chosen, can be expressed as
\begin{equation}
    \dspec = !\left(\frac{\dis^2 L}{2}\right) \times {L \choose \dis^2 L/ 2}, \quad \dis = \sqrt{\frac{2m}{L}}, \quad m = 2,\,3,\,\cdots,\,L, 
    \label{eq:dis_spec_perm}
\end{equation}
where $!n$ is the number of derangments of an $n$-element set \cite{hassani03}, which can be expressed as
\begin{equation}
    !n = \left\lfloor \frac{n!}{e} + \frac{1}{2} \right\rfloor.
\end{equation}
Note that the distance spectrum is independent of the particular signal choice in this example. 

Substituting \eqref{eq:dis_spec_perm} into \eqref{eq:newUB1}, the new upper bound on the error probability of the permutations based coding can be expressed as
\begin{equation}
\begin{aligned}
P_e^{NB} = \min_{\gamma \geq 0} \Biggl\{\sum_{j=1}^N b_j (1-e^{-\gamma/\lambda_j}) + \frac{1}{\pi}\sum_{\dhalf =2}^{L}\sum_{j=1}^{N}!\dhalf {L \choose \dhalf}  &\int_{0}^{\pi/2}\frac{2b_j L \sigma^2 \sin^2\theta}{\dhalf E\lambda_j + 2 L \sigma^2 \sin^2\theta} \\ & \times\exp{\left(\frac{\dhalf E\lambda_j + 2 L \sigma^2 \sin^2\theta}{2\lambda_j L \sigma^2 \sin^2\theta}\gamma \right)}\diff \theta \Biggl\}.
\end{aligned}
\end{equation}


Fig. \ref{fig:perm} shows the block error probability performance versus average received SNR for the permutations based coding scheme with $L \in \{3,\, 6,\, 9\}$. We adopt the correlated Rayleigh fading channel model with $N = 2$ and $\rho = 0.1$. As we can observe from Fig. \ref{fig:perm}, the new bound is tighter than the union bound in the full range of SNRs. When $L = 9$, the union bound is completely uninformative as its value is greater than $1$ even in the high SNR region up to 28dB, while our proposed bound is always lower than $1$. The SNR gap between the new bound and the traditional union bound increases from around 0dB to 15dB when increasing $L$ from $3$ to $9$. 
This indicates that the union bound is far more significantly affected by the $\mathcal{O}(L!)$ growth in the number of PEPs.
\begin{figure}[t]
    \centerline{\includegraphics[width=10cm,height=7.8cm]{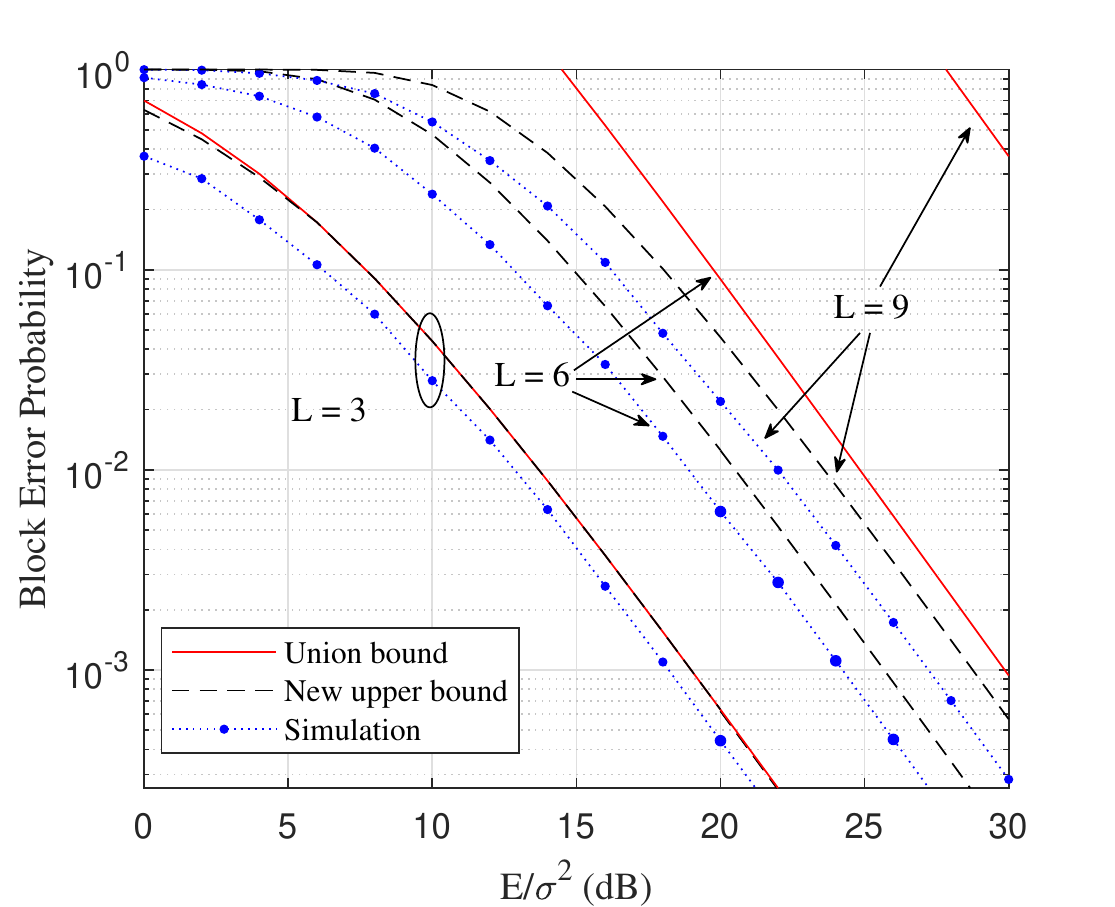}}
    \caption{The block error probability versus average received SNR for permutations based coding with $L \in \{3,\, 6,\, 9\}$ in a correlated Rayleigh fading channel with $N = 2$ and $\rho = 0.1$.
    }\label{fig:perm}
\end{figure}

\subsection{Gaussian distribution based random coding}  \label{sec:UBGaus}
Random coding is a common method for proof of theorems in information theory \cite{Cover05}. Although it is not a practical signalling method, the average error probability calculated from random code selection shows the existence of good practical codes. 
We consider the example of randomly generated codes with the same dimension, $\dimension$, but different number of symbols, $\numcodes$. We assume that the $\numcodes$ signals are generated following a circularly symmetric Gaussian distribution as
\begin{equation}
\signal_i \sim \mathcal{CN}\left(0, \boldsymbol{I}_\dimension \right), \, i \in \{1,\, \cdots,\, \numcodes\}.
\end{equation}
The signals are then normalised such that $\frac{1}{\numcodes}\sum_{i=1}^{\numcodes} \|\signal_i\|^2 = 1$. These randomly generated signals form a constellation diagram with arbitrary decision regions that are hard to evaluate.

Fig. \ref{fig:gau} provides the block error probability versus average received SNR for Gaussian distribution based random coding with $\dimension = 9$ and $\numcodes \in \{10,\, 300\}$. We adopt the correlated Rayleigh fading channel model with $N = 2$ and $\rho = 0.1$. 
\begin{figure}[t]
    \centerline{\includegraphics[width=10cm,height=7.8cm]{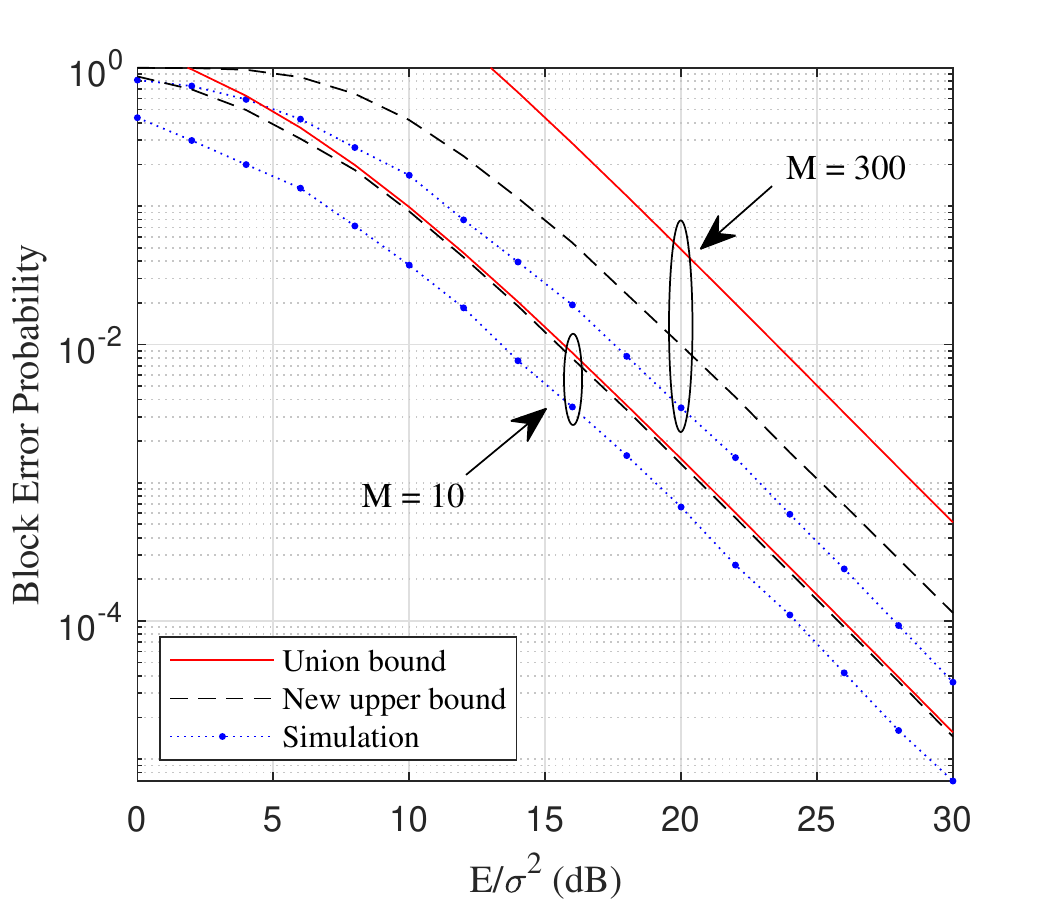}}
    \caption{The block error probability versus average received SNR for Gaussian distribution based random coding with $\dimension = 9$ and $\numcodes \in \{10,\, 300\}$ in a correlated Rayleigh fading channel with $N = 2$ and $\rho = 0.1$.
    }\label{fig:gau}
\end{figure}
Since the signals are randomly generated, it is difficult to provide a concise expression for the distance spectrum $\dspec$. Instead, we numerically calculate $\dspec$ for the randomly generated codewords and substitute it into \eqref{eq:newUB} in order to evaluate the proposed bound in Fig. \ref{fig:gau}. It can be clearly observed that in the high SNR region, the SNR gap between the proposed bound and the traditional upper bound increases from $0.2$dB to $3.5$dB when increasing $\numcodes$ from $10$ to $300$. This again emphasises the prominence of the increasing number of PEPs in the union bound and the advantage of the proposed bound when the number of pairwise error events is large.

\section{Conclusions} \label{sec:conclusion}
This paper presents an improved union bound for maximum-likelihood detection based error probability performance over fading channels. In contrast to existing work  on tightening the union bound, we focus on the deep fading events, based on which the domain of the random variable associated with the channel gain is separated into two parts. The upper bound is then tightened by assuming that a pairwise error always occurs in a deep fading event. The resulting expression provides a more meaningful upper bound when compared to the union bound which can produce an error probability bound above $1$, especially in situations with a large number of PEPs. It is important to note that the objective function of the optimisation problem required in the bound has a single minimium and can be easily solved using numerical methods.  

While the new bound is presented for any fading channel model, to gain further insight, we derive an easy-to-evaluate expression under the correlated Rayleigh model and provide several interesting examples to illustrate the improved performance of the proposed bound.





%

\appendices
\section{Proof of Theorem \ref{theo:property}}  \label{app:proof} 
We first take the derivative of $G(\gamma)$ with respect to $\gamma$, which results in
\begin{equation}
\frac{\diff G(\gamma)}{ \diff{ \gamma}} 
= f_X(\gamma)\left(1-\frac{1}{\numcodes}\sum_{i=1}^{\numcodes}\sum_{\dis \in \dset} \dspec Q\left( \dis\sqrt{\frac{E \gamma}{2\sigma^2}} \right) \right). \label{eq:newUB_1stdiff} 
\end{equation}
Note that \eqref{eq:newUB_1stdiff} becomes zero when either $f_X(\gamma)$ or $\left(1-\frac{1}{\numcodes}\sum_{i=1}^{\numcodes}\sum_{\dis \in \dset} \dspec Q\left( \dis\sqrt{\frac{E \gamma}{2\sigma^2}} \right) \right)$ is zero. 

Focusing on the second term of \eqref{eq:newUB_1stdiff}, since $\lim_{x \rightarrow \infty} Q(x) = 0$, we can show that 
\begin{equation}
    \lim_{\gamma \rightarrow \infty} \left(1-\frac{1}{\numcodes}\sum_{i=1}^{\numcodes}\sum_{\dis \in \dset} \dspec Q\left( \dis\sqrt{\frac{E \gamma}{2\sigma^2}} \right)\right) = 1 > 0. \label{eq:Ginf}
\end{equation} 
Then, by using the assumption that $\numcodes \geq 3$, we obtain
\begin{equation}
1-\frac{1}{\numcodes}\sum_{i=1}^{\numcodes}\sum_{\dis \in \dset} \dspec Q\left( 0 \right) \leq 0, \label{eq:G0}
\end{equation} 
since $Q(0) = 0.5$ and $\sum_{\dis \in \dset} \dspec = K-1$. 

Since $Q\left( \dis\sqrt{\frac{E \gamma}{2\sigma^2}} \right)$ is monotonically decreasing with $\gamma$, $\left(1-\frac{1}{\numcodes}\sum_{i=1}^{\numcodes}\sum_{\dis \in \dset} \dspec Q\left( \dis\sqrt{\frac{E \gamma}{2\sigma^2}} \right)\right)$ is monotonically increasing with $\gamma$. By combining \eqref{eq:Ginf}, \eqref{eq:G0} and the monotonically increasing property, we can ensure that there exists only a single $\gamma^* \in [0,\, \infty)$ such that 
\begin{equation}
1-\frac{1}{\numcodes}\sum_{i=1}^{\numcodes}\sum_{\dis \in \dset} \dspec Q\left( \dis\sqrt{\frac{E \gamma^*}{2\sigma^2}} \right) = 0   \label{eq:gamma_star}
\end{equation}
Since $f_X(0) \geq 0$, another possible stationary point can be introduced by the first term of \eqref{eq:newUB_1stdiff} as $\gamma = 0$. Thus, our objective function has two possible stationary points. 

Next, we proceed to take the second derivative in order to show that $\gamma^*$ gives a minimum. The second derivative of $G(\gamma)$ with respect to $\gamma$ can be expressed as
\begin{equation}
\begin{aligned}
\frac{\diff^2 G(\gamma)}{ \diff{ \gamma^2}} = \frac{\diff f_X(\gamma)}{ \diff{ \gamma}} + \frac{1}{\numcodes}\sum_{i=1}^{\numcodes}\sum_{\dis \in \dset} \dspec \left( \frac{1}{\sqrt{2\pi}}\frac{\dis^2 E}{4 \sigma^2}\exp\left(-\frac{\dis^2 E \gamma}{4 \sigma^2}\right)\left( \dis\sqrt{\frac{E \gamma}{2\sigma^2}} \right)^{-1/2}  f_X(\gamma)  \right.\\
\left. - Q\left( \dis\sqrt{\frac{E \gamma}{2\sigma^2}} \right)\frac{\diff f_X(\gamma)}{ \diff{ \gamma}} \right). \label{eq:newUB_2nddiff}
\end{aligned}
\end{equation}
Evaluating \eqref{eq:newUB_2nddiff} at $\gamma = \gamma^*$ and using the result in \eqref{eq:gamma_star} gives
\begin{equation}
\frac{\diff^2 G(\gamma^*)}{ \diff{ \gamma^2}} = \frac{1}{\numcodes}\sum_{i=1}^{\numcodes}\sum_{\dis \in \dset} \dspec \frac{1}{\sqrt{2\pi}}\frac{\dis^2 E}{4 \sigma^2}\exp\left(-\frac{\dis^2 E \gamma}{4 \sigma^2}\right)\left( \dis\sqrt{\frac{E \gamma^*}{2\sigma^2}} \right)^{-1/2}  f_X(\gamma^*) \geq 0, \label{eq:newUB_2nddiff1}
\end{equation}
which indicates that $\gamma^*$ minimises $G(\cdot)$. If $f_X(0)=0$, the only other stationary point is $\gamma = 0$ and must corresponds to a maximum. Therefore, $G(\gamma)$ has a single minimum at $\gamma = \gamma^*$.




\ifCLASSOPTIONcaptionsoff
  \newpage
\fi



%



\bibliographystyle{IEEEtran}

\end{document}